\documentclass[10pt,twocolumn,twoside,aps,pra,nopacs,superscriptaddress]{revtex4-1}
\usepackage{graphicx}
\usepackage{amsmath,amssymb}
\usepackage{color}
\usepackage{bbm}
\usepackage{hyperref}
\usepackage{theorem}
\usepackage{latexsym}

\newtheorem{theorem}{Theorem}

\newtheorem{definition}[theorem]{Definition}
\newtheorem{example}[theorem]{Example}

\newtheorem{proposition}[theorem]{Proposition}
\newtheorem{remark}[theorem]{Remark}

\newlength{\blank}
\newenvironment{proof}[1][{\hspace{-\blank}}]{{\noindent\textbf{Proof~{#1}.\ }}}{\hfill\qed}

\mathchardef\ordinarycolon\mathcode`\:
\mathcode`\:=\string"8000
\def\vcentcolon{\mathrel{\mathop\ordinarycolon}}
\begingroup \catcode`\:=\active
  \lowercase{\endgroup
  \let :\vcentcolon
  }

\newcommand{\nc}{\newcommand}
\nc{\rnc}{\renewcommand}
\nc{\beq}{\begin{equation}}
\nc{\eeq}{{\end{equation}}}
\nc{\beqa}{\begin{eqnarray}}
\nc{\eeqa}{\end{eqnarray}}
\nc{\lbar}[1]{\overline{#1}}
\nc{\ketbra}[2]{|#1\rangle\!\langle#2|}
\nc{\proj}[1]{| #1\rangle\!\langle #1 |}
\nc{\avg}[1]{\langle#1\rangle}
\nc{\Rank}{\operatorname{Rank}}
\nc{\smfrac}[2]{\mbox{$\frac{#1}{#2}$}}
\nc{\tr}{\operatorname{Tr}}
\nc{\ox}{\otimes}
\nc{\dg}{\dagger}
\nc{\dn}{\downarrow}
\nc{\cA}{\mathcal{A}}
\nc{\cB}{\mathcal{B}}
\nc{\cC}{\mathcal{C}}
\nc{\cD}{\mathcal{D}}
\nc{\cE}{\mathcal{E}}
\nc{\cF}{\mathcal{F}}
\nc{\cG}{\mathcal{G}}
\nc{\cH}{\mathcal{H}}
\nc{\cI}{\mathcal{I}}
\nc{\cJ}{\mathcal{J}}
\nc{\cK}{\mathcal{K}}
\nc{\cL}{\mathcal{L}}
\nc{\cM}{\mathcal{M}}
\nc{\cN}{\mathcal{N}}
\nc{\cO}{\mathcal{O}}
\nc{\cP}{\mathcal{P}}
\nc{\cR}{\mathcal{R}}
\nc{\cS}{\mathcal{S}}
\nc{\cT}{\mathcal{T}}
\nc{\cX}{\mathcal{X}}
\nc{\cZ}{\mathcal{Z}}
\nc{\csupp}{{\operatorname{csupp}}}
\nc{\qsupp}{{\operatorname{qsupp}}}
\nc{\var}{\operatorname{var}}
\nc{\rar}{\rightarrow}
\nc{\lrar}{\longrightarrow}
\nc{\polylog}{\operatorname{polylog}}
\nc{\1}{{\openone}}
\nc{\id}{{\operatorname{id}}}

\nc{\RR}{{{\mathbb R}}}
\nc{\CC}{{{\mathbb C}}}
\nc{\FF}{{{\mathbb F}}}
\nc{\NN}{{{\mathbb N}}}
\nc{\ZZ}{{{\mathbb Z}}}
\nc{\PP}{{{\mathbb P}}}
\nc{\QQ}{{{\mathbb Q}}}
\nc{\UU}{{{\mathbb U}}}
\nc{\EE}{{{\mathbb E}}}

\nc{\qed}{{$\hfill\Box$}}

\begin{document}

\title{What does an experimental test of quantum contextuality prove or disprove?}
 
\author{Andreas Winter}
\email{andreas.winter@uab.cat}
\affiliation{ICREA \&{} F\'{\i}sica Te\`{o}rica: Informaci\'{o} i Fen\`{o}mens Qu\`{a}ntics, %
Universitat Aut\`{o}noma de Barcelona, ES-08193 Bellaterra (Barcelona), Spain}

\date{10 August 2014}

\begin{abstract}
The possibility to test experimentally the Bell-Kochen-Specker theorem
is investigated critically, following the demonstrations by Meyer,
Kent and Clifton-Kent that the predictions of quantum mechanics are
indistinguishable (up to arbitrary precision) from those of a non-contextual 
model, and the subsequent debate to which extent these models are actually
classical or non-contextual.

The present analysis starts from a careful consideration these
``finite-precision'' approximations. A stronger condition for 
non-contextual models, dubbed \emph{ontological faithfulness}, 
is exhibited.
It is shown that it allows to formulate approximately the constraints in 
Bell-Kochen-Specker theorems such as to render the usual proofs
robust. As a consequence, one can experimentally test to finite
precision ontologically faithful non-contextuality, and thus experimentally
refute explanations from this smaller class.
We include a discussion of the relation of ontological faithfulness
to other proposals to overcome the finite precision objection.
\end{abstract}


\maketitle


%
%

%
%

Attempts at proving that quantum mechanics 
is fundamentally non-classical go back to its very beginning. 
They start with the Copenhagen school's claims of the 
necessity of complementarity~\cite{Bohr:compl,Heisenberg:uncert}, 
to von Neumann's proof~\cite{vonNeumann:no-go}, via
Einstein-Podolsky-Rosen's attempt to show the incompleteness of 
quantum theory~\cite{EPR}, and on to an intellectual culmination in the 
work of Bell~\cite{Bell64,Bell66}, building on Gleason~\cite{Gleason}, 
as well as Kochen and Specker~\cite{KS67,Specker1960}. 
But the subject is full of vigour even today
as we witness a renewed evaluation of the foundations of quantum
mechanics~\cite{PR,ColbeckRenner,PBR,finite-speed,Almost-quantum}, now even 
reaching out to applications such as device independent quantum
cryptography~\cite{MasanesPironioAcin,HaenggiRenner}.

The work of Bell especially marked a turning point for the discourse
on the non-classical nature of quantum mechanics, in two ways: On
the one hand it showed the need for clear, operationally motivated 
criteria for the classicality of a theoretical explanation; on the 
other hand it demonstrated operational differences between the quantum
mechanical predictions and those of any theory based on classical 
hidden variable in the former sense. This opened the way for
experimental tests of quantum mechanics at an unprecedented level,
from quantum violations of Bell inequalities~\cite{CHSH,Aspect,Weihs}
to experimental verification of quantum 
contextuality~\cite{Cabello,Simon-et-al,Michler-et-al,Kirchmair-et-al,Ahrens-et-al}.

But just as these various no-go theorems have inspired the 
thinking of physicists, so have their refutations, or 
\emph{nullifications}, to use the term of Meyer~\cite{Meyer:null}.
Here we shall take the view of no-go theorems as actually true 
mathematical theorems, 
hence a refutation, rather than being the demonstration of
a mathematical error, consists in showing that a tacit, yet not
necessarily plausible assumption was made.

The present paper concerns the Bell-Kochen-Specker (BKS) theorems
on the impossibility of a non-contextual hidden variable
explanation of the predictions of quantum mechanics. 
We shall essentially stick to the original viewpoint of these 
authors, owing to Bell's 1966 review~\cite{Bell66} (which appeared
in print later, but actually predates the 1964 non-locality paper~\cite{Bell64},
whose semi-centenary is celebrated in the present special issue);
non-contextuality is a property of hidden variable theories,
assigning ``true'' outcomes to all or some observables in a 
quantum mechanical system, relative to a given state. 
Risking to labour an obvious point, recall that for the
quantum theory it does not matter how these observations are made,
as long as they lead to a well-defined POVM
-- and there will always be very different-looking procedures,
involving separate systems, different quantum information carriers,
or temporal orderings in which the outcome is generated step-by-step.
Regarding the hypothetical hidden variable description, however, 
assumptions have to be made regarding the relation of the variables
attached to different or even incompatible observables, or relating them
to different measurement procedures. Locality of the hidden variables
is one of these conditions~\cite{Bell64,CHSH}, assuming a multi-partite
quantum system with space-like separated observers. Non-contextuality
can be viewed as a more ``bare bones'' and abstract condition,
in particular not requiring a multi-partite system.
It is important to observe that the assumption of locality imposes
non-contextuality on certain sets of hidden variables, a fact
noticed repeatedly (for a recent and exhaustive discussion see
e.g.~\cite{Specker-s-parable}). Conversely, Bell-Kochen-Specker
proofs of quantum contextuality have inspired Bell inequalities
for non-locality~\cite{Cabello:Belllll}.

Since its inception, the notion of non-contextuality has been
deeply examined in a variety of forms. In particular the work
of Spekkens~\cite{Spekkens} is noteworthy, in that it identified
several distinct aspects of contextuality, and because it argued the
necessity to have operational, theory-independent definitions of
the basic terms. The recent paper~\cite{ChiribellaYuan}
has taken up this challenge for the notion of ``sharp measurement''.
(Non-)contextuality is also a recent hot topic in the ongoing quest
for physical axioms limiting the range of conceivable probabilistic
theories, in particular 
\emph{Local Orthogonality}~\cite{Fritz-et-al}
and \emph{Consistent Exclusivity}~\cite{Cabello-E,Henson,AFLS,FLS,Yan,Amaral-et-al},
were anticipated by Specker~\cite{Specker1960}, cf.~\cite{Cabello:Speckerrrrr}.

\medskip
The present paper is structured as follows: Section~\ref{sec:BKS} gives an 
account of usual (infinite precision) BKS theorems in terms of
non-contextual inequalities, followed by the objections
raised by Meyer, Kent and Clifton-Kent (MKC) in section~\ref{sec:MKC}
due to finite precision of any realistic experiment.
After that we present some reflections on the necessity of being
able to identify outcomes in different experiments 
(section~\ref{sec:outcome-identity}) as being in a certain sense 
``the same''. 
The central part is section~\ref{sec:precision}, where the notion of
\emph{ontologically faithful non-contextuality (ONC)} is
introduced, which is designed to reflect the finite precision 
of actual experiments in the supposed hidden variable theory; we show
then that this notion gives substance to experimental tests
of quantum contextuality, and concretely making non-contextual
inequalities robust to finite precision.
In section~\ref{sec:conclusions} we discuss our results and make
a comparison with other approaches, in particular one based on
sequential execution of measurements.

\section{Bell-Kochen-Specker Theorems \protect\\ and Non-Contextual Inequalities}
\label{sec:BKS}
Traditionally, Bell-Kochen-Specker (BKS) proofs start from a collection
$(P_i:i \in V)$ of projectors on a finite-dimensional Hilbert
space. Each subset $M \subset V$ of indices with the property
that $\sum_{i\in M} P_i = \1$, describes a measurement, more precisely
a von Neumann measurement, with possible outcomes labelled by
$M$. Given a quantum state $\rho$, the Born rule then prescribes
the probabilities of outcomes for each measurement $M$ that
can be formed by collecting projectors $P_i$:
\begin{equation}
  \label{eq:Born_rulz}
  \Pr\{i|\rho\} = \tr\rho P_i.
\end{equation}
The remarkable thing is that this probability depends only on the
outcome $i$, more precisely on the projector $P_i$ (given the state,
which we consider fixed in our discussion).

\emph{A classical non-contextual hidden variable model} for this scenario
is meant to reveal pre-existing values of the measurements, but in
such a way that the indicator $X_i \in \{0,1\}$ depends on the label $i$ only,
and not whether $P_i$ appears in a measurement $M$ or in another 
measurement $M'$ -- the different contexts. 
Here, $X_i = 1$ indicates that $i \in M$ is the outcome 
if $M$ was measured. The model reproduces the predictions of quantum
theory if $\EE X_i = \Pr\{i|\rho\} = \tr\rho P_i$ in accordance 
with the Born rule (\ref{eq:Born_rulz}).

Clearly, a necessary requirement for the possibility of such a model
is that for all possible measurements $M\subset V$, 
i.e.~$\sum_{i\in M} P_i = \1$,
\begin{equation}
  \label{eq:BKS}
  \sum_{i\in M} X_i = 1.
\end{equation}
In words: the assignment $X_i$ picks one and only one ``real'' 
outcome for each measurement $M$. To be painfully precise, the
relation~(\ref{eq:BKS}) should hold with probability $1$ (allowing
for inequality on an event of probability $0$), but we can ignore 
this detail for the present discussion.
Recall that a random variable is a (measurable) function 
$X_i:\Omega \longrightarrow \RR$ from a probability space
$\Omega$ with a probability distribution $\mu$ (where there is
implicit the $\sigma$-algebra of events; cf.~\cite{Feller} or
any other modern textbook on probability for the basic terms).
Since $X_i$ takes only values $0$ and $1$, it is equivalently
described by the sets (\emph{events} in probability jargon)
\[
  E_i = X_i^{-1}(1) = \{ \omega\in\Omega : X_i(\omega) = 1 \},
\]
so that $X_i = 1_{E_i}$ is the indicator function of $E_i$.
The condition~(\ref{eq:BKS}) is then equivalent to 
\begin{equation}
  \label{eq:BKS-sets}
  \dot{\bigcup_{i\in M}} E_i = \Omega,
\end{equation}
where the notation on the left hand side refers to the disjoint union.
In other words, $(E_i:i\in M)$ is a set partition of $\Omega$.

Kochen and Specker~\cite{KS67} and later many other authors, for a 
selection see~\cite{PeresMermin,Cabello-18} and references therein,
have found sets of projectors $P_i$ in three- and higher-dimensional
Hilbert spaces for which these conditions are contradictory: there
exists no non-contextual hidden-variable model satisfying
either of eqs.~(\ref{eq:BKS}) or (\ref{eq:BKS-sets}). 

In those works, the concept of \emph{colouring} is central, which is the
evaluation of the $X_i$ on a single point $\omega\in\Omega$:
$X_i(\omega) \in \{0,1\}$ indicates whether $i$ is singled
out or not. The colouring rule is that in every measurement $M$
one must choose one and only one element. Note that a $0$-$1$-colouring
is a special case of our above notion of classical non-contextual
hidden variable model, namely when the $X_i$ take values $0$ or
$1$ with unit probability. Conversely, a probability distribution
over colourings gives rise to a classical non-contextual hidden variable model.

With respect to these colourings, we strongly suggest however 
to take a probabilistic point of view as explained above, since
the aim of a hidden variable theory is not merely a logically 
consistent assignment of values (surely a necessary condition),
but the explanation of observed statistical data. 
Indeed, a broader approach, leaving a role also for the
quantum state, is to consider within the set $V$ subsets $C$
(which we shall call \emph{contexts}) such that $\sum_{i\in C} P_i \leq \1$;
i.e., rather than demanding that $C$ describes a measurement, we
only require that it can be completed to one. The set of all contexts,
denoted $\Gamma$, is a collection of subsets of $V$ (aka \emph{hypergraph}),
and now a non-contextual hidden variable model only has to satisfy
\begin{equation}
  \label{eq:BKS-inequality}
  \forall \text{ contexts } C\in\Gamma \quad \sum_{i\in C} X_i \leq 1.
\end{equation}
With this small modification 
every set of projectors, with associated collection 
$\Gamma$ of contexts, has a non-contextual hidden variable model,
but there are differences in the attainable expectation values
$(t_i=\Pr\{i\}:i\in V)$, be it as quantum expectations
$(\tr\rho P_i:i\in V)$ or as classical expectations $(\EE X_i:i\in V)$. 

For instance, in a traditional BKS proof, where $\sum_{i\in C} P_i = \1$
for each context $C\in\Gamma$, we know that 
$\sum_{i\in C} \EE X_i \leq 1$, but we cannot reach
equality in all of the $C\in\Gamma$, hence
\begin{equation}
  \sum_{C\in\Gamma} \sum_{i\in C} \EE X_i \leq |\Gamma|-1.
\end{equation}
By contrast, quantum mechanics attains 
$|\Gamma| = \sum_{C\in\Gamma} \sum_{i\in C} \langle P_i \rangle$, 
in fact for \emph{every} state!
Thus, BKS proofs can be interpreted as statements concerning
the (im-)possibility of realising certain constraints among 
projector effects in the quantum case, and events in the 
classical case.

More generally, we may consider \emph{non-contextual inequalities}
of the form
\begin{equation}
  \label{eq:nc-inequality}
  \sum_{i\in V} \lambda_i \EE X_i \leq \beta_{\text{cl}},
\end{equation}
where $\lambda_i \geq 0$ are certain coefficients and $\beta_{\text{cl}}$
is the maximum of the l.h.s.~over all non-contextual hidden variable
models, i.e.~$X_i\in\{0,1\}$ satisfying eq.~(\ref{eq:BKS-inequality}).
Substituting quantum expectations, 
$\sum_i \lambda_i \tr \rho P_i$ may exceed the classical limit $\beta_{\text{cl}}$
up to a quantum maximum of $\beta_{\text{qu}}$. A logically consistent
way to think about these structures is to \emph{start} with the
set $V$ of outcomes and the contexts, i.e.~with the hypergraph
$\Gamma$ -- this sets the scene of the possible experiments we want
to describe and their outcomes (crucially identifying outcomes
in different experiments as being the same; we'll return to this
later). 
In this way we can divorce the logic of speaking about the 
experiments from the theory that we believe or hypothesize to
underly.
For instance, a classical, non-contextual, model for this abstract
structure is given by $0$-$1$-variables $X_i$ satisfying
(\ref{eq:BKS-inequality}), while a quantum model is a collection
of projectors $P_i$ with
\begin{equation}
  \label{eq:quantum-model}
  \forall C\in\Gamma \quad \sum_{i\in C} P_i \leq \1.
\end{equation}
Note that classical models for $\Gamma$ are a special case
of quantum models, where all $P_i$ commute.
See~\cite{CabelloSeveriniWinter10}, where this
formalism was developed further (including also more general
probabilistic theories with convex sets of states and linear
functions on states as effects). In~\cite{AFLS} this framework
was developed even more, however returning to the equality 
conditions (\ref{eq:BKS}) and $\sum_{i\in C} P_i = \1$.
To justify the emphasis on linear inequalities, observe that
\begin{align*}
  \mathcal{E}(\Gamma) &= \bigl\{ (\EE X_i) : (X_i) \text{ non-contextual HV model} \bigr\}, \\
  \Theta(\Gamma)      &= \bigl\{ (\tr\rho P_i) : (P_i) \text{ quantum model, } 
                                                                       \rho \text{ state} \bigr\},
\end{align*}
are convex, the first being in fact a polytope, cf.~\cite{Knuth,CabelloSeveriniWinter10,AFLS}.

\begin{remark}
  As we are dealing with projectors,
  the condition $\sum_{i\in C} P_i \leq \1$ is evidently equivalent
  to $P_i P_j = 0$ (i.e.~orthogonality of their supports) for
  all $i\neq j$ occurring jointly in \emph{some} context $C \ni i,j$. 
  
  Likewise, for a classical, non-contextual model, the condition
  $\sum_{i\in C} X_i \leq 1$ is equivalent to $X_i X_j = 0$ for
  all $i\neq j$ occurring jointly in \emph{some} context $C \ni i,j$.
  
  This relation defines a graph on $V$, with an edge
  $i\sim j$ if and only if there is a $C\in\Gamma$ with $i,j\in C$.
  It is known as 
  \emph{exclusivity graph~\cite{AFLS,CabelloSeveriniWinter14}}.
\end{remark}

\medskip
Rather than dwelling more on the abstract formalism, let us
look at an example~\cite{Klyachko}, which is indeed the one that
inspired the general hypergraph approach~\cite{CabelloSeveriniWinter10}:
\begin{example} 
  \label{ex:klyachko}
  \normalfont
  For $5$ outcomes $V=\{0,1,2,3,4\}$ and contexts 
  $\Gamma = \{ 01,12,23,34,40 \}$, there is a well-known quantum
  realization by $5$ rank-one projectors $P_i$ in three-dimensional
  Hilbert space: This means that $P_i P_{i+1} = 0$ for all $i$,
  where $i+1$ is understood $\mod 5$. Notice that from these projectors
  many POVMs can be built, the simplest ones being the binary 
  measurements $(P_i,\1-P_i)$, but as this and $(P_{i+1},\1-P_{i=1})$
  are compatible, we also have $(P_i,P_{i+1},\1-P_i-P_{i+1})$.
  
  However, there can be many
  other quantum models, including higher-rank projectors in
  higher dimension. For a suitably chosen state $\rho$
  (and interestingly not any state), we can achieve
  $\sum_i \tr\rho P_i = \sqrt{5}$~\cite{Klyachko,Lovasz79}, which is indeed 
  the quantum maximum,
  $\beta_{\text{qu}} = \sqrt{5}$~\cite{CabelloSeveriniWinter10}.
  
  On the other hand, it is straightforward to check that 
  the maximum of $\sum_{i=0}^4 \EE X_i$ is $\beta_{\text{cl}}=2$.
\end{example}

\medskip
The above translation, from an unsatisfiable set of logical
constraints to a limitation on expectation values, all of which
are in principle observable quantities, is significant.
It elevates (and generalizes) BKS proofs to experimentally
testable propositions, or so it would seem: 
Indeed, non-contextual hidden variable models impose a bound 
$\beta_{\text{cl}}$ on the expectation value of $\sum_i \lambda_i X_i$, 
eq.~(\ref{eq:nc-inequality}), which is violated by the
quantum expectation value $\sum_i \lambda_i \langle P_i \rangle$.
There is only one catch, or rather the very reason why quantum
mechanics can outperform non-contextual classical models: The
quantum expectation values are not accessible in a single
von Neumann measurement. What is more, it is necessary for a
gap $\beta_{\text{cl}} < \beta_{\text{qu}}$ to occur, that
the same projector $P_i$ occurs in different, incompatible
measurements, for various $i$.
In the next section we shall see that this poses not only a
conceptual problem, but also a practical one when purportedly testing
quantum contextuality (i.e.~experimentally refuting classical
non-contextual hidden variable explanations).

\section{Meyer-Clifton-Kent's Nullification of BKS}
\label{sec:MKC}
Let us start with an easy objection against any physical relevance 
of BKS theorems stemming from the fact that the hidden variable theory 
for a set of effects is supposed to assign pre-existing values only
to projective measurements. But in experiments it is highly unlikely
that ever a sharp von Neumann measurement is implemented. What is
more, experimental evidence based on observable expectation values
never allows the experimenter to distinguish conclusively between
a projector (an element in an ideal measurement) and some arbitrarily
close POVM element (aka effect), i.e.~a positive semidefinite operator 
upper bounded by $\1$, be it another projector or a genuinely non-projective
POVM element. 
However, in this form this does not pose a worry, since the quantum
mechanical expectations values of two operators $A$ and $B$ cannot differ 
by more than $\|A-B\|$, the operator norm of their difference. This 
norm difference can be experimentally estimated via the fundamental
relation
\[
  \| A-B \| = \max_{\rho\text{ state}} | \tr\rho A - \tr\rho B |.
\]
Thus, we can at least in principle confirm experimentally
(within the rules of quantum mechanics and according to our
command of the underlying physical system) 
that the experiments implement 
POVM elements close to the required projectors. 
%
%
This is important because
allowing general POVM elements, one can reach values of
$\sum_i \lambda_i \langle P_i \rangle$ even larger than
$\beta_{\text{qu}}$, all the way to
\[
  \beta_{\text{g}} = \max_i \lambda_i t_i
                     \quad \text{s.t.}\ \forall i\ 0\leq t_i\leq 1,\ 
                                        \forall C\in\Gamma\ \sum_{i\in C} t_i \leq 1,
\]
the maximum value allowable by generalised probabilistic
theories~\cite{Wright,CabelloSeveriniWinter10}. For instance
for the pentagon (Example~\ref{ex:klyachko}), this value is
$\beta_{\text{g}} = \frac52$.

At the same time, we would naturally demand that whatever the 
experiment does, it should have a classical hidden variable 
explanation for each measurement outcome. 

This brings us to the objection by Meyer~\cite{Meyer:null}, which was
greatly refined and extended by Kent~\cite{Kent:null}, 
as well as Clifton and Kent~\cite{CliftonKent:null}. These authors
show that in each dimension $d$ of the underlying Hilbert space
of a quantum system, there exists a dense set 
$\mathfrak{M} = \{ M^{(1)},M^{(2)}, \ldots \}$ (w.l.o.g.~countable)
of complete von Neumann measurements $M^{(j)} = (P_{j1},\ldots,P_{jd})$,
consisting of rank-one projectors, with the property that every
$P_{jk}$ occurs in only one measurement, namely $M^{(j)}$.
Here, ``dense'' refers to the set of all von Neumann measurements:
for every von Neumann measurement $(Q_1,\ldots,Q_d)$ and every
$\epsilon > 0$, there exists an $M^{(j)} \in \mathfrak{M}$ such that
$\| Q_k - P_{jk} \| \leq \epsilon$ for all $k$.

This set $\mathfrak{M}$ of measurements clearly has a non-contextual
hidden variable model reproducing the correct statistics for any given
state $\rho$: Any random variables $X_{jk} \in \{0,1\}$ such that
\[
  \Pr\{ X_{jk}=1,\ X_{j\widehat{k}} = 0\ \forall \widehat{k}\neq k \}
                                                      = \tr \rho P_{jk}
\]
will do.
[In fact, this can even be extended to POVMs with a bounded
number of outcomes.]
Such sets are not so hard to come by, either by existence proofs
or constructively. 

Now, if the experimenter needs to implement a measurement
$Q = (Q_1,\ldots,Q_d)$, she can only ensure (and demonstrate by
experimental verification) that she has done so up to a finite
accuracy $\epsilon$. In particular, she cannot distinguish her
experimental observations from those of a suitably close
measurement $M^{(j)} \in \mathfrak{M}$ -- which however has
a genuinely non-contextual hidden variable explanation! In 
practice, hence, where one can never be sure which one of the 
infinitely many measurements arbitrarily close to $Q$ was responsible
for the observations, the experimenter cannot rule out a fully
non-contextual hidden variable theory, nor the concomitant
requirement that ``really'' only the measurements in the set
$\mathfrak{M}$ are implemented. 
Note that this has nothing to do with the correctness of the 
mathematical reasoning of Gleason, Bell, Kochen-Specker, and 
so on (of which simply the prerequisites do not apply), but 
only concerns its relevance to the observable world.

Not surprisingly perhaps, this argument, though simple and
in our opinion unrefutable, has sparked a considerable debate
that continues in some form until today, but a 
review of which is beyond the scope of this paper. See however
the excellent discussion and extensive references in the article
of Barrett and Kent~\cite{BarrettKent:null}. 

In a nutshell, it is clear to physicists that \emph{something} quantum 
is demonstrated in the experiments: An inequality is violated
and incompatible measurements are performed in successive runs.
Yet, the Meyer-Clifton-Kent (MKC) argument 
shows clearly that something is lacking to
be able to claim experimental confirmation of quantum contextuality.
So, what can we salvage from this intolerable situation?

\section{What Does ``Same Outcome'' Mean?}
\label{sec:outcome-identity}
All discussions of hidden variables in quantum mechanics have 
to labour one point, in one way or another: the identification
of the quantum theoretic entities, as given by the formalism,
with counterparts in the hypothetical hidden variable model.
This is where inevitably unproven assumptions are made -- have 
to be made, indeed, as we know next to nothing about this 
alternative description, which after all may not even be real,
and rather only want to talk about
the general form of that hypothetical theory, which in addition
usually is only meant to ``explain'' some small section of the
range of observable phenomena. 

In non-contextual hidden variable models, this has to be
done on the level of the measurement outcomes by postulating
$0$-$1$-indicator random variables, singling out at most one 
of the possible outcomes of each measurement (one and only one 
in the original BKS theorems, but as explained in 
section~\ref{sec:BKS} we can relax this to ``at most one''). 
So far, this is nothing peculiar, and already the end of the story
in the MKC models (see the previous section). The curious thing,
however, is that in addition we demand that certain outcomes ``$i$''
in one measurement and outcome ``$j$'' in another measurement
(which may require completely different experiments one from
the other) are to be identified. For the hidden variable theory
this is taken to mean that the two associated random variables
take the same value; for the quantum model it means that 
these outcomes are represented by the same projector, more general
same effect, in the different measurements. 

Operationally we can recognize identical effects by their 
identical response to all different state preparations. 
Turning this around, this is how a theory tells us what its 
effects are: concise descriptions of the different
responses (as probabilities of a ``click'') of experiments. 
Cf.~the theory of Ludwig~\cite{Ludwig}, in which states are
equivalence classes of state preparations with respect to the
statistics under all possible measurements; and vice versa
measurements (and indeed effect) equivalence classes of experiments
with respect to the statistics under all possible state
preparations. In particular, coarse graining the outcomes of
any procedure to only two yields a binary measurement, for
example assigning ``yes'' if a particular click happens, and
``no'' otherwise. If starting from two experiments and 
singling out an outcome in each of them, we end up with equivalent
binary observables, this allows us to identify the effects
corresponding to these outcomes. (Recall that a priori these
outcomes are simply defined as clicks of a certain kind in two
potentially completely different experimental procedures.)

Following Spekkens' operational approach~\cite{Spekkens}, 
a non-contextual classical
hidden variable model has to assign random variables to each
effect, in particular the same value to all appearances of that
effect in different contexts. There is probably not any compelling 
reason to believe a priori that the hidden variables should 
have this property; but it is undoubtedly a very intuitive idea,
rooted in our classical intuition, cf.~Bell's discussion~\cite{Bell66}.

What the MKC constructions exploit is that our access to the
quantum mechanical effects is governed by a topology of
closeness, rather than identity; on the other hand, the previous
requirement of non-contextuality does not put any conditions
on the classical entities assigned to two distinct, but
infinitesimally close effects.
Towards the end of Spekkens' article~\cite{Spekkens}, it is suggested 
that to overcome the MKC critique based on finite precision, we need an 
extended version of the above identification principle, requiring that
the elements representing similar effects in the quantum
theory are in some sense similar. In the next section we present
a simple proposal to formulate such a similarity principle.

\section{Finite Precision: \protect\\ Ontological Faithfulness}
\label{sec:precision}
The effects of quantum theory have a geometry thanks to the
operator norm $\| A-B \|$: As explained in section~\ref{sec:MKC}, 
this is the largest difference between expectation
values of the two operators on the same state. Experiments
consequently can pin down an effect only up to statistical
error bars in this norm. 
In particular, let us consider a set $\cE$ of effects sufficiently 
close to projectors. We call this an \emph{$\epsilon$-precise quantum model} 
of a set $V$ of outcomes $i$ and a collection $\Gamma$ of
contexts $C\subset V$, if every $Q_i\in \cE$ is $\epsilon$-close
to a projector $P_i$, and furthermore for every $C\in\Gamma$,
there is a collection of $Q^C_i\in\cE$ such that
$\sum_{i\in C} Q_i^C \leq \1$ and $\|Q_i^C-P_i\|\leq\epsilon$
for all $i\in C$. This encapsulates the notion of contexts $C$,
each outcome $i$ of which can be approximately identified with 
a projector $P_i$.

The corresponding classical hidden variable models are captured
by the following definition.

\begin{definition}
  \label{def:epsilon-ONC}
  An \emph{$\epsilon$-ontologically faithful non-contextual 
  ($\epsilon$-ONC) model} for a hypergraph $\Gamma$ of
  contexts $C \subset V$ consists of a family of random variables 
  $X^{C}_{i} \in \{0,1\}$, $i\in C \in \Gamma$, such that
  \begin{align*}
    &\forall C\in\Gamma \quad \sum_{i\in C} X^{C}_{i} \leq 1, \text{ and}\\
    &\forall C,C'\in\Gamma\ \forall i\in C\cap C' \quad
                       \Pr\{ X^{C}_{i} \neq X^{C'}_{i} \} \leq \epsilon.
  \end{align*}

  In other words: For each context $C\in\Gamma$, the family
  $(X^{C}_{i}:i\in C)$ is a classical hidden variable model of a
  measurement containing the outcomes $C$; and the models
  for the ``same'' outcome $i$ occurring in different contexts 
  $C$ and $C'$ almost coincide.
\end{definition}

\medskip
The spirit of this definition is that it imposes a distance on 
the random variables $X^C_i$ representing the different
``incarnations'' of the outcome $i$. As we are dealing with 
$0$-$1$-variables, the probability of disagreeing is a natural
such measure, but in more complex situations other distances
may be employed.

Logically, both definitions, of an approximate quantum model and
of an ontologically faithful non-contextual classical model, are
independent, hinging directly on the combinatorial structure
of the permissible contexts $\Gamma$. However, it may be helpful to
think of the underlying idea as a robustification of the 
BKS assignment of a $0$-$1$-random variable to each projector in 
a set of quantum measurements -- see Fig.~\ref{fig:NC-ONC}.

\begin{figure}[ht]
  \includegraphics[width=8.4cm]{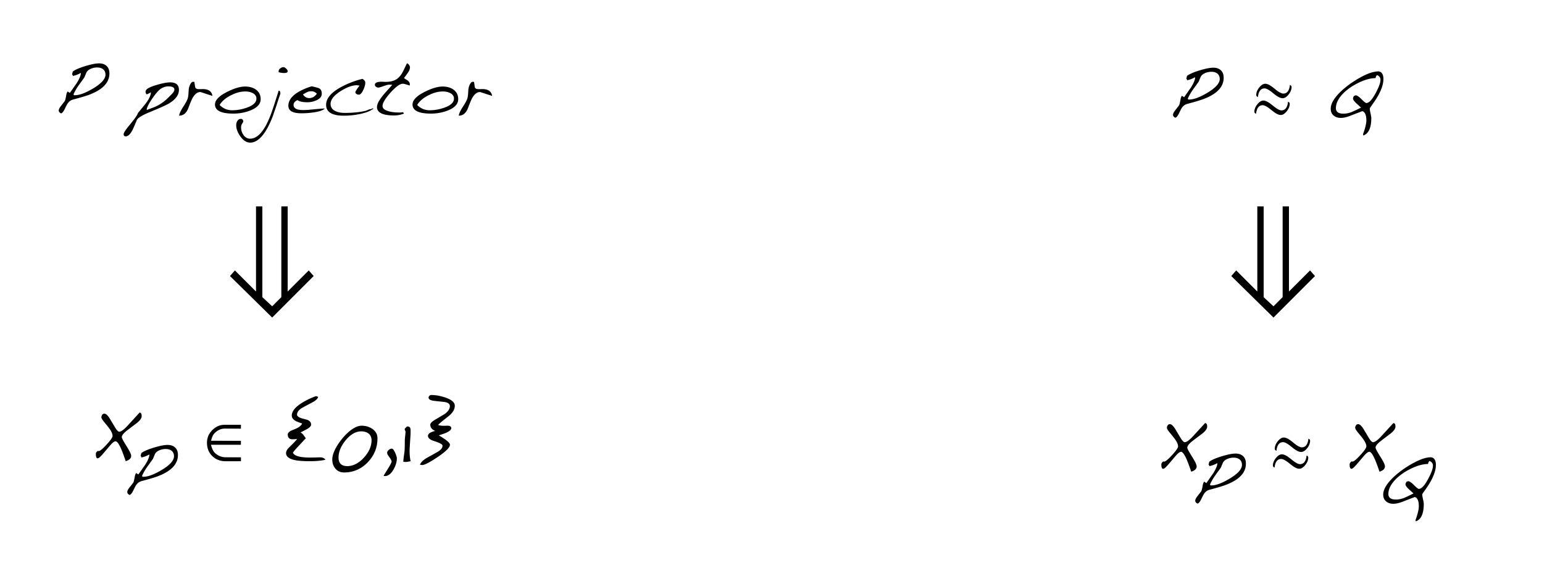}
  \caption{Left: A non-contextual hidden variable model assigns a random variable
           $X_P\in\{0,1\}$ to projectors $P$. 
           Right: In an ontologically faithful non-contextual model, 
           this assignment has to respect the geometry of the space 
           of projectors as well as that of the random variables. On
           the quantum side, it is given by the operator norm; on the hidden
           variable side, by the probability of being distinct.}
  \label{fig:NC-ONC}
\end{figure}

\medskip\noindent
{\bf From approximate to exact non-contextual hidden variables.}
Now we show how the $\epsilon$-approximations introduced above
can be eliminated at the expense of having, with some small probability, 
no outcome among the set $V$.
Indeed, we can easily build a non-contextual hidden variable model from
any $\epsilon$-ONC model by letting
\begin{equation}
  \label{eq:epsilon-to-zero}
  Y_i := \prod_{i\in C\in\Gamma} X^{C}_i.
\end{equation}
Note that $Y^{C}_i := Y_i$ for $i\in C$ defines a $0$-ONC
model and that it approximates the original one well:
\begin{proposition}
  \label{prop:epsilon-to-zero}
  Consider an $\epsilon$-ONC model $(X^C_i)_{i\in C\in\Gamma}$ 
  and associated $Y_i$ as per eq.~(\ref{eq:epsilon-to-zero}).
  Let $k_i$ be the number of times an outcome $i$ occurs in some
  context $C\in\Gamma$. Then the collection $(Y_i)$ is a non-contextual 
  hidden variable model, and for every $i\in V$,
  \[
    \Pr\bigl\{ \exists C\ni i\ \text{s.t.}\ X^C_i \neq Y_i \} \leq (k_i-1)\epsilon.
  \]
\end{proposition}
\begin{proof}
  To check that $(Y_i)$ defines a non-contextual hidden variable
  model, we observe for every $i\in C\in \Gamma$, $Y_i \leq X^C_i$, hence
  \[
    \sum_{i\in C} Y_i \leq \sum_{i\in C} X^C_i \leq 1.
  \]
  On the other hand, for each $i$,
  \[\begin{split}
    \Pr&\bigl\{ \exists C\ni i\ \text{s.t.}\ X^C_i \neq Y_i \}          \\
       &\phantom{==}
        =    \Pr\bigl\{ (X^C_i\,:\,i\in C\in\Gamma) \ \text{not all equal} \} \\
       &\phantom{==}
        \leq (k_i-1)\epsilon,
  \end{split}\]
  because it is enough to compare a single $X^{C_0}_i$ against
  each other $X^C_i$.
\end{proof}

We record especially the implication of the preceding result for
(linear) non-contextual inequalities.

\begin{proposition}
  \label{prop:ONC-to-NC}
  Consider a non-contextual inequality
  $\sum_i \lambda_i \EE X_i \leq \beta_{\text{cl}}$ for a set $\Gamma$
  of contexts $C \subset V$, and $k_i$ as in Proposition~\ref{prop:epsilon-to-zero}.
  
  Then, for any $\epsilon$-ONC $(X^C_i)$ and any assignment of a context
  $i \longmapsto C_i \ni i$, and letting $t_i := \EE X^{C_i}_i$ we have
  \[
    \sum_i \lambda_i t_i \leq \beta_{\text{cl}} + \epsilon \sum_i \lambda_i(k_i-1).
  \] 
\end{proposition}
\begin{proof}
  Let $Y_i := \prod_{i\in C\in\Gamma} X^{C}_i$ as in eq.~(\ref{eq:epsilon-to-zero}).
  Then, Proposition~\ref{prop:epsilon-to-zero} implies
  \[
    t_i = \EE X^{C_i}_i \leq \EE Y_i + (k_i-1)\epsilon,
  \]
  and summing over $i$ we are done.
\end{proof}

\begin{example}
  \normalfont
  Returning to the pentagon inequality of Klyachko 
  \emph{et al.}~\cite{Klyachko} (Example~\ref{ex:klyachko}),
  we have $k_i=2$ for all outcomes $i=0,1,2,3,4$.
  The original inequality as $\sum_i \EE X_i \leq 2 = \beta_{\text{cl}}$,
  i.e.~all $\lambda_i = 1$. Thus, for any 
  \[
    \epsilon < \frac{\sqrt{5}-2}{5} \approx 0.047,
  \]
  the observation of a value sufficiently close to 
  $\beta_{\text{qu}} = \sqrt{5}$ in a quantum setup with projectors
  rules out $\epsilon$-ONC hidden variables.
\end{example}

\begin{example}
  \normalfont
  Another example is the contextuality proof via the \emph{Mermin-Peres
  square}~\cite{PeresMermin}, which is also the subject of the
  experiment in~\cite{Kirchmair-et-al}. The set of outcomes one has
  to consider here consists of $24$ rank-one projectors in $\CC^4$, 
  forming $24$ contexts of four elements each, each of which corresponds
  to a complete orthonormal basis (cf.~\cite[Fig.~1]{ent+e=0}).
  
  It can be shown that $\sum_i \EE X_i \leq 5 = \beta_{\text{cl}}$,
  whereas the maximum quantum values is $\beta_{\text{qu}}=6$
  (once more see~\cite{ent+e=0}). Since each outcome occurs in exactly
  $k_i=4$ contexts, a value of
  \[
    \epsilon < \frac{6-5}{72} \approx 0.0138
  \]
  suffices to rule out $\epsilon$-ONC hidden variables for this
  experiment.
\end{example}

\medskip
Clearly, in the same way every non-contextual inequality is made robust
to finite precision, by restricting it to $\epsilon$-ONC models
with suitably small $\epsilon > 0$. In particular, this means 
that the MKC hidden variable models cannot be ontologically
faithful, since they are consistent with measuring arbitrary quantum
observable to any finite precision.

As a consequence, MKC-like hidden variable models necessarily have to
have parameters inaccessible in quantum mechanics, since they have
to assign very different classical random variables to arbitrarily
close projectors.

\section{Discussion}
\label{sec:conclusions}
Any proof that classical hidden variables
cannot reproduce the predictions of quantum mechanics is
invariably achieved only under some other assumption, be it
non-contextuality or locality of the classical variables. The
assumption itself is not testable, so it has to be chosen on
other, ``reasonable'', grounds. For instance, in Bell tests, one
would argue that no-signalling is a well-established fact (it
follows from special relativity) and that non-local classical 
variables would have to come with some incredible mechanism
to remain absolutely hidden to prevent some eventual faster-than-light
signal getting out.

For non-contextuality the case is a priori weaker, as we are not
taking recourse to another physical principle (compare however
the ideas formulated in~\cite{Cabello-mem,CHHH} regarding memory
bounds). Instead, we make
assumptions on how to describe experiments. That there are
such things as experiments can hardly be denied, the setup
subjecting a system in some reproducible preparation to measurement,
yielding a result. In addition, quantum mechanics already comes
with the identification between outcomes in experiments and 
effect operators, even up to finite precision and statistical
error bars. However it cannot, by definition, tell us anything 
about the hypothetical hidden-variable theory -- especially if
the no-go theorem involves exhibiting an operational difference
between the two.

Here, we have shown how to include finite precision into the reasoning
about non-contextual hidden variable theories. Our analysis
revolves around the idea that both in quantum,
as well as in possible classical models of any given structure
of contexts on an abstract set of outcomes, this requires 
introducing a metric on the entities of the model reflecting
the degree of approximation. The definition of ontological
faithfulness is one, presumably however not the only, way to 
formalize such a notion.

\medskip
Since the MKC argument has been put forward, and following the 
subsequent debate, other attempts to address 
the finite precision objection have appeared. 
Ignoring the ones that aimed at finding a flaw of some kind in 
MKC (see~\cite{BarrettKent:null} for an extensive review),
we find rather more interesting those introducing some other 
additional property of the hidden variables that should guarantee
experimental testability. The most developed of these is in the
papers of Cabello \emph{et al.}~\cite{Cabello-et-al},
Kirchmair \emph{et al.}~\cite{Kirchmair-et-al} and 
G\"uhne \emph{et al.}~\cite{Guehne-et-al} (cf.~also~\cite{Larsson-et-al}); 
there, the new element
of sequential measurements was introduced and the assumption 
concerns the behaviour of the hidden parameter, and the observable
consequences thereof, under sequences of (almost) compatible 
measurements. Some such assumption is necessary as one can see from
a suitable extension of the MKC models to include a simulation
of the state change due to measurement, which then is able to
reproduce to arbitrary precision the statistics of sequences
of measurements and the effects of the projection postulate.

It may not be possible to make a full comparison between this
proposal and the one presented here, as ontological faithfulness
is a minimally invasive change of the BKS approach in which
any experiment is an indivisible whole. In particular, it does not
imply anything about the history or dynamics of the hidden parameter,
as~\cite{Cabello-et-al,Kirchmair-et-al,Guehne-et-al,Larsson-et-al} 
and related approaches necessarily do.
However, it is worth noting that simply assuming no back-action 
on the hidden parameter, in an $\epsilon$-ONC, we have for $i \in C \cap C'$,
\begin{equation}
  \Pr\{ X_i^{C'} \neq \xi | X_i^C = \xi \} \leq \frac{\epsilon}{\Pr\{X_i^C=\xi\}}.
\end{equation}
I.e., unless the probability of observing an
outcome $i$ in a context $C$ is very small, the subsequent
consultation of the ``same'' outcome in a different context $C'$
yields the same value, with high probability. Thus, ontological
faithfulness of the hidden variable theory implies an
approximate version of repeatability of measurements, even 
on the level of the same outcome in different contexts.

Conversely, it seems likely that the approach of~\cite{Guehne-et-al,Larsson-et-al}
always implies an $\epsilon$-ONC hidden variable model, 
because the assumptions in those papers imply that one can define
hidden variables corresponding to measuring some observable
jointly with other, commuting, ones, and these different 
variables turn out to be $\epsilon$-close in the sense of
Definition~\ref{def:epsilon-ONC}.
For this purpose, we stress that sequential application
of several measurement devices yields just --- another measurement,
as does another ordering of the same devices, or the unitary transfer 
of the state into multiple qubits which are subsequently measured
in space-like separation, or for that matter any magical mystery machine as
long as it is governed by the rules of quantum mechanics. What counts
are the observations that are actually going to be made, and once that 
is decided, these observations can be reflected in a 
suitable compatibility structure of outcomes and contexts. 
The quantum mechanics of these devices and their 
sequential application is taken care of by quantum theory itself; 
not so for the hypothetical hidden variable theory, which has to be 
augmented by additional assumptions to allow for any meaningful
comparison with quantum mechanics.

\medskip
We shall refrain from arguing the plausibility of ontological
faithfulness; that does not seem to make a lot of sense
to the present author, unless one actually believes in 
hidden variables (contextual or non-contextual)
as a viable description of quantum reality.
What after all is the point of no-go theorems? In the best case they 
reveal the incompatibility between a set of preconceptions on the one
hand and a certain description of nature on the other. This
surely is a worthy enterprise, especially since we keep struggling 
with that description in the case of quantum mechanics. In this vein, 
no-go theorems, rather than demonstrating a lack of imagination, really are 
indispensable tools for a fruitful use of it~\cite{Bell82}.

\bigskip
\acknowledgments
Conversations with Jonathan Barrett, Ad\'{a}n Cabello, Giulio
Chiribella, Jonathan Oppenheim, Simone Severini and Rob Spekkens 
on the topic of Bell-Kochen-Specker contextuality proofs are gratefully 
acknowledged. 
Special thanks are due to Sabrina Merli and Enza Di Tomaso for
providing the space-time window in which the core parts of the
present research were carried out.
At that time, the author was affiliated
with the School of Mathematics, University of Bristol, and the
Centre for Quantum Technologies, National University of Singapore.
He was supported by the U.K. EPSRC, the Royal Society,
the Philip Leverhulme Trust,
the European Commission (STREPs ``QCS'' and ``RAQUEL''), 
the ERC (Advanced Grant ``IRQUAT''), 
and the Spanish MINECO (grant number FIS2008-01236) with FEDER funds.

\end{document}